\documentclass[9pt]{article}%
\usepackage{amssymb}
\usepackage{spconf}
\usepackage{amsmath}
\usepackage{graphicx}
\usepackage{algorithm}
\usepackage{amsthm}
\usepackage[english]{babel}%
\usepackage{amsfonts}

\providecommand{\U}[1]{\protect\rule{.1in}{.1in}}

\newtheorem{lemma}{Lemma}

\name{M. Rossi,$^{\ast }$ A. M. Tulino,$^{\dag }$ O. Simeone,$^{\ast }$\thanks{The work of O. Simeone was partially supported by the U. S. National Science Foundation under Grant \# CCF-0914899.} and A. M. Haimovich $^{\ast }$}
\address{$^{\ast }$CWCSPR, New Jersey Institute of Technology, Newark, NJ 07102, USA\\
$^{\dag }$Bell Laboratories (Alcatel-Lucent), Holmdel, NJ 07733, USA.\\
{\small{\texttt{\{mr227,osvaldo.simeone,alexander.m.haimovich\}@njit.edu},}}\\
{\small{\texttt{a.tulino@alcatel-lucent.com}}}}
\input{tcilatex}

\begin{document}

\title{Non-Convex Utility Maximization in Gaussian\\MISO Broadcast and Interference Channels}
\author{}
\maketitle

\begin{abstract}
\ninept Utility (e.g., sum-rate) maximization for multiantenna broadcast and
interference channels (with one antenna at the receivers) is known to be in
general a non-convex problem, if one limits the scope to linear (beamforming)
strategies at transmitter and receivers. In this paper, it is shown that,
under some standard assumptions, most notably that the utility function is
decreasing with the interference levels at the receivers, a global optimal
solution can be found with reduced complexity via a suitably designed
branch-and-bound method. Although infeasible for real-time implementation,
this procedure enables a non-heuristic and systematic assessment of suboptimal
techniques. In addition to the global optimal scheme, a real-time suboptimal
algorithm, which generalizes the well-known distributed pricing techniques, is
also proposed. Finally, numerical results are provided that compare global
optimal solutions with suboptimal (pricing) techniques for sum-rate
maximization problems, affording insight into issues such as the robustness
against bad initializations in real-time suboptimal strategies.

\end{abstract}

\ninept

\begin{keywords}
Nonconvex optimization, branch-and-bound, interference channel, multiple-input single-output channel
\end{keywords}

\section{Introduction}

\label{sec:intro}

Precoding and power control are well studied strategies that support high
spectral efficiency in wireless network with multiple antenna transceivers,
when channel state information (CSI) is available at the transmitters. Several
system-wide objective functions have been considered in the literature for
precoding and power control optimization of broadcast channels (BCs) and
interference channels (ICs). Some of these problems are convex, for example
power minimization \cite{bib: Trans Power} or SINR balancing for the
multiple-input single-output (MISO) BC \cite{bib: boche sinr}, and thus
solvable with standard techniques in reasonable (polynomial) time. However, in
general, the problems at hand are non-convex. Unlike convex problems,
non-convex problems typically do not afford efficient (i.e., polynomial-time)
algorithms that are able to achieve global optimality \cite{bib: boyd}. For
example, it is known that the weighted sum-rate maximization (WSRM) in
parallel IC channels, where interference from other users is treated as noise
(a non-convex problem) is NP-hard \cite{bib: complexity and duality} (this
result extends also to BC as a special case).

In this paper, we address the global minimization of a system-wide, in general
non-convex, cost function with respect to the transmit covariance matrices
$\left\{  \mathbf{Q}_{k}\right\}  $. Among global techniques, branch-and-bound
(BB)\ algorithms are methods to solve general\ non-convex problems \cite{bib:
boyd 1991}, producing an $\varepsilon$-suboptimal feasible point. BB methods
have been already introduced to solve non-convex power control problems,
although so far only multi-user single-input single-output systems have been
addressed in \cite{bib: ngoc} and references therein.

In this paper, we propose a novel BB framework for global optimization of a
problem formulation that includes, for instance, MISO BC and IC WSRM with
general convex power constraint. The proposed BB approach is based on the
observation that a fairly general set of cost functions that arise in
communication's problems, albeit non-convex, possess a \textit{Partly
Convex-Monotone} \cite{bib: tuy decomposition}\ structure. This structure is
satisfied whenever one can identify a suitable set of interference functions
$\mathbf{\mathbf{f}}_{\mathbf{i}}\left(  \left\{  \mathbf{Q}_{k}\right\}
\right)  $, for which the following hold: (\textit{i}) The cost function is
\textit{convex} in the transmit covariance matrices $\left\{  \mathbf{Q}%
_{k}\right\}  $ once the interference functions $\mathbf{\mathbf{f}%
}_{\mathbf{i}}\left(  \left\{  \mathbf{Q}_{k}\right\}  \right)  $ are fixed;
(\textit{ii}) The cost function is \textit{monotone} in the interference
functions $\mathbf{\mathbf{f}}_{\mathbf{i}}\left(  \left\{  \mathbf{Q}%
_{k}\right\}  \right)  $. We design the BB scheme to exploit the
\textit{Partly Convex-Monotone} structure of the problem. Branching is
performed in a reduced space (of the size of the set of all feasible
interference level vectors $\mathbf{\mathbf{f}}_{\mathbf{i}}\left(  \left\{
\mathbf{Q}_{k}\right\}  \right)  $), instead of the original
feasible$\mathcal{\ }$space (of the size of the set of all feasible covariance
matrices $\left\{  \mathbf{Q}_{k}\right\}  $).\ Bounding is efficiently
carried out by solving only convex optimization problems.

In addition to the reduced-space BB method, we propose a suboptimal algorithm
that attains quasi-optimal performance with polynomial complexity. This
algorithm reduces to the distributed pricing scheme of \cite{bib: berry honig
pricing MISO}, when applied to sum-rate maximization problems. Numerical
results are provided to compare the global optimal solution based on BB, the
suboptimal (pricing) technique and the nonlinear dirty-paper coding scheme.

\textit{Notation}: The Boldface is used to denote matrices (uppercase) and
vectors (lowercase); $\left(  \cdot\right)  ^{T}$ and $\left(  \cdot\right)
^{H}$ denote the transpose and the Hermitian transpose, respectively;
$\operatorname{Tr}\left(  \cdot\right)  $ denotes the trace of a matrix;
$E\left[  \cdot\right]  $ denotes the expectation operator. Moreover, given a
vector $\mathbf{x}$ we address its $l$-th component as $\left[  \mathbf{x}%
\right]  _{l} $, and the vector inequality $\mathbf{x}\preccurlyeq
\mathbf{\mathbf{\mathbf{y}}}$ means that $\left[  \mathbf{x}\right]  _{l}%
\leq\left[  \mathbf{y}\right]  _{l}$ $\forall l$. Finally, unless otherwise
specified, we address the set of covariance matrices $\left\{  \mathbf{Q}%
_{k}\right\}  _{k=1}^{K}$ as $\left\{  \mathbf{Q}_{k}\right\}  $.

\section{System Model\label{system_model}}

We model a multi-user communication system consisting of $K$
transmitter-receiver pairs (or users). The $k$-th user has $N_{k}$\ transmit
antennas and one receive antenna (MISO system). The signal at the $k$-th
receiver is given by%
\begin{equation}
y_{k}=\underset{\text{\textit{signal}}}{\underbrace{\mathbf{h}_{kk}%
\mathbf{x}_{k}}}+\underset{\text{\textit{interference}}}{\underbrace
{\sum\nolimits_{j\neq k}\mathbf{h}_{jk}\mathbf{x}_{j}}}+w_{k}\label{eq: model}%
\end{equation}
where $\mathbf{x}_{k}\in\mathbb{C}^{N_{k}\times1}$ is the $k$-th transmitter's
signal, $\mathbf{h}_{kj}\in\mathbb{C}^{1\times N_{k}}$ accounts for the
channel response of the MISO link between the $k$-th transmitter and the
$j$-th receiver, and $w_{k}\in\mathbb{C}^{1\times1}$ models the additive white
Gaussian noise (AWGN) at $k$-th receiver: $w_{k}\sim\mathcal{CN}(0,\sigma
_{k}^{2})$. Assuming capacity-achieving Gaussian codebooks, we define the
correlation matrix of the $k$-th transmitted signal as $\mathbf{Q}%
_{k}=E\left[  \mathbf{x}_{k}\mathbf{x}_{k}^{H}\right]  $. While model
(\ref{eq: model}) accounts for an IC, a BC can be obtained as a special case
by setting $\mathbf{h}_{jk}=\mathbf{h}_{k}$\ $\forall j,k$.

\subsection{Problem Formulation\label{Sec_optimal_linear_precoder}}

Due to multi-user interference, the system performance depends on the
transmission strategy of every user, i.e., on the set of covariance matrices
$\left\{  \mathbf{Q}_{k}\right\}  $. We consider the minimization with respect
to $\left\{  \mathbf{Q}_{k}\right\}  $ of a system-wide cost function $f$ (to
be defined below) under a\ general convex set constraints $\mathcal{Q}$:%
\begin{equation}
\underset{\left\{  \mathbf{Q}_{k}\right\}  \in\mathcal{Q}}{\min}f\left(
\left\{  \mathbf{Q}_{k}\right\}  ,\mathbf{\mathbf{f}}_{\mathbf{i}}\left(
\left\{  \mathbf{Q}_{k}\right\}  \right)  \right) \label{eq: f(Q_W)}%
\end{equation}
By defining a set of $L$ auxiliary variables $\mathbf{i}$, problem
(\ref{eq: f(Q_W)}) can be recast in the equivalent form%
\begin{align}
\left(  \mathrm{P}\right)  \text{\hspace{0.1in}}\min_{\left\{  \mathbf{Q}%
_{k}\right\}  \in\mathcal{Q},\mathbf{i}}f\left(  \left\{  \mathbf{Q}%
_{k}\right\}  ,\mathbf{i}\right)   & \label{eq: Problem_auxiliar_i}\\
s.t.\text{\ \ }\mathbf{i}={\mathbf{f}}_{\mathbf{i}}\left(  \left\{
\mathbf{Q}_{k}\right\}  \right)   & \nonumber
\end{align}
The equivalence means that if $\left\{  \mathbf{Q}_{k}^{\ast}\right\}  $ is a
solution to (\ref{eq: f(Q_W)}), then $\left(  \left\{  \mathbf{Q}_{k}^{\ast
}\right\}  ,\mathbf{\mathbf{f}}_{\mathbf{i}}\left(  \left\{  \mathbf{Q}%
_{k}^{\ast}\right\}  \right)  \right)  $ is a solution to
(\ref{eq: Problem_auxiliar_i}). Conversely, if $\left(  \left\{
\mathbf{Q}_{k}^{\ast}\right\}  ,\mathbf{\mathbf{i}}^{\mathbf{\mathbf{\ast}}%
}\right)  $ is a solution to (\ref{eq: Problem_auxiliar_i}), then $\left\{
\mathbf{Q}_{k}^{\ast}\right\}  $ is a solution to (\ref{eq: f(Q_W)}).\newline
We further make the following assumptions:

\begin{enumerate}
\item[A1] The $L$ interference levels are given by the real vector function
$\mathbf{\mathbf{f}}_{\mathbf{i}}\left(  \left\{  \mathbf{Q}_{k}\right\}
\right)  $, \textit{affine}\ with respect to\ $\left\{  \mathbf{Q}%
_{k}\right\}  $, that is bounded in the $L$-dimensional rectangle $\left[
\mathbf{i}^{\min},\mathbf{i}^{\max}\right]  \subset\mathbb{R}^{L}$ (i.e., the
$l$-th component satisfies $i_{l}^{\min}\leq\left[  \mathbf{\mathbf{f}%
}_{\mathbf{i}}\left(  \left\{  \mathbf{Q}_{k}\right\}  \right)  \right]
_{l}\leq i_{l}^{\max} $ for $l=1,\ldots,L$). For instance, we typically have
$L=K$ and the interference level at the $k$-th receiver reads $\left[
\mathbf{\mathbf{f}}_{\mathbf{i}}\left(  \left\{  \mathbf{Q}_{k}\right\}
\right)  \right]  _{k}=\sum\nolimits_{j=1,j\neq k}^{K}\mathbf{\mathbf{h}}%
_{jk}\mathbf{Q}_{j}\mathbf{\mathbf{h}}_{jk}^{H}$;

\item[A2] The cost function $f\left(  \left\{  \mathbf{Q}_{k}\right\}
,\mathbf{i}\right)  $ is a real scalar function that is: \textit{continuous}
in $\left(  \left\{  \mathbf{Q}_{k}\right\}  ,\mathbf{i}\right)  $;
\textit{monotonic increasing}\footnote{By suitably modifying the same
arguments, the proposed framework can handle an analogous but more general
case where $f\left(  \left\{  \mathbf{Q}_{k}\right\}  ,\mathbf{i}^{\mathbf{+}%
},\mathbf{i}^{\mathbf{-}}\right)  $ results \textit{monotone increasing} in
$\mathbf{i}^{\mathbf{+}}$ and \textit{monotone decreasing} in $\mathbf{i}%
^{\mathbf{-}}$.} with respect to\ $\mathbf{i}\in\left[  \mathbf{i}^{\min
},\mathbf{i}^{\max}\right]  $ for\ fixed $\left\{  \mathbf{Q}_{k}\right\}
\in\mathcal{Q}$; \textit{convex} with respect to $\left\{  \mathbf{Q}%
_{k}\right\}  $ for fixed $\mathbf{i}\in\left[  \mathbf{i}^{\min}%
,\mathbf{i}^{\max}\right]  $;

\item[A3] The set $\mathcal{Q}$ is \textit{closed} and \textit{convex}. For
example, $\mathcal{Q}$ may be the set of positive semidefinite covariance
matrices $\left\{  \mathbf{Q}_{k}\succcurlyeq\mathbf{0}\right\}  $ satisfying
the generalized power constraints $\sum\nolimits_{k=1}^{K}\operatorname*{Tr}%
\left(  \mathbf{A}_{k,\ell}\mathbf{Q}_{k}\right)  \leq P_{\ell}$ for
$\ell=1,\ldots,D$, where $\left\{  \mathbf{A}_{k,\ell}\right\}  $ are positive
semidefinite matrices (possibly $\mathbf{A}_{k,\ell}=\mathbf{0}$ if $k$-th
user doesn't belong to $\ell$-th constraint) and $\left\{  P_{\ell}\right\}  $
are non-negative coefficients. This definition includes some important special
cases studied in the literature, such as per-antenna, per-group of antennas,
the classical sum-power or the interference constraints in cognitive radio scenarios.
\end{enumerate}

Throughout the paper, we refer to problem (\ref{eq: Problem_auxiliar_i}) as
$\left(  \mathrm{P}\right)  $. We next provide examples of\ problems that
satisfy these assumptions.

\subsection{Examples}

\label{marker: example Utility}

An example of cost function included in our framework is the $\alpha$-fairness
criterion \cite{bib: alfa fairness}: $f\left(  \left\{  \mathbf{Q}%
_{k}\right\}  \right)  =\sum_{k=1}^{K}-w_{k}f_{\alpha}\left(  r_{k}\left(
\left\{  \mathbf{Q}_{k}\right\}  \right)  \right)  $, where $w_{k}$ is a
positive constant, $f_{\alpha}$\ is an increasing strictly concave function
defined as%
\begin{equation}
f_{\alpha}\left(  r\right)  \colon=\left\{
\begin{array}
[c]{lc}%
\log r & \text{if }\alpha=1\\
\left(  1-\alpha\right)  ^{-1}r^{1-\alpha} & \text{otherwise}%
\end{array}
\text{,}\right.
\end{equation}
and $r_{k}\left(  \left\{  \mathbf{Q}_{k}\right\}  \right)  $ is the $k$-th
user's rate, which depends on covariance matrices $\left\{  \mathbf{Q}%
_{k}\right\}  $ and on the \textit{channel scenario}. The $\alpha$-fairness
criterion reduces, as special cases, to the WSRM problem ($\alpha=0$) or
the\ \textit{proportional fairness} problem ($\alpha=1$). Moreover, as
$\alpha$ becomes large, it converges to the\ \textit{max--min fairness}
problem \cite{bib: alfa fairness}.\newline In the following we present some
examples of channel scenario that can be addressed within our framework:

\begin{itemize}
\item \textit{Parallel MISO IC}: The $k$-th transmitter operates over $L_{C}$
parallel subcarriers, it has power constraint $P_{k}$ and has knowledge of
channels $\left\{  \mathbf{\mathbf{h}}_{jkl}\right\}  $ for $j=1,\ldots,K$ and
$\forall l$.\newline The minimization of the $\left(  p,\alpha\right)
$-fairness cost function reads%
\begin{gather}
\underset{\left\{  r_{k}\right\}  ,\left\{  \mathbf{Q}_{kl}\succcurlyeq
\mathbf{0}\right\}  }{\min}\sum_{k=1}^{K}-w_{k}f_{\alpha}\left(  r_{k}\right)
\label{eq: IC WSRM}\\
s.t.\left\{
\begin{array}
[c]{ll}%
r_{k}\leq\sum\limits_{l=1}^{L_{C}}\log\left(  1\mathbf{+}\frac
{\mathbf{\mathbf{h}}_{kkl}\mathbf{Q}_{kl}\mathbf{\mathbf{h}}_{kkl}^{H}}%
{\sigma_{kl}^{2}\mathbf{+}%
{\textstyle\sum_{j=1,j\neq k}^{K}}
\mathbf{\mathbf{h}}_{jkl}\mathbf{Q}_{jl}\mathbf{\mathbf{h}}_{jkl}^{H}}\right)
& \forall k\\
\operatorname{Tr}\left(  \sum_{l=1}^{L_{C}}\mathbf{Q}_{kl}\right)  \leq P_{k}
& \forall k
\end{array}
\right. \nonumber
\end{gather}
Defining\ $\left[  \mathbf{\mathbf{f}}_{\mathbf{i}}\left(  \left\{
\mathbf{Q}_{kl}\right\}  \right)  \right]  _{l+L_{C}\left(  k-1\right)  }%
=\sum\nolimits_{j=1,j\neq k}^{K}\mathbf{\mathbf{h}}_{jkl}\mathbf{Q}%
_{jl}\mathbf{\mathbf{h}}_{jkl}^{H}$\ \ $\forall k,l$ and $\mathcal{Q}=\left\{
\mathbf{Q}_{kl}\succcurlyeq\mathbf{0}\text{ }\forall k,l\text{\ }|\text{
}\operatorname{Tr}\left(  \sum_{l=1}^{L_{C}}\mathbf{Q}_{kl}\right)  \leq
P_{k}\text{\ }\forall k\right\}  $, problem (\ref{eq: IC WSRM}) is recast
into\ $\left(  \mathrm{P}\right)  $. Also, $\mathbf{\mathbf{f}}_{\mathbf{i}%
}\left(  \left\{  \mathbf{Q}_{kl}\right\}  \right)  \in\left[  \mathbf{0,i}%
^{\max}\right]  $ where $\mathbf{i}^{\max}$ is a proper upper bound on
interference, always available since $\mathcal{Q}$ is bounded (finite power constraints).

\item \textit{Parallel MISO\ BC}: This scenario is obtained from
(\ref{eq: IC WSRM}) by setting $\mathbf{h}_{jkl}=\mathbf{h}_{kl}$\ $\forall j$
and imposing a sum-power constraint $\operatorname{Tr}\left(  \sum_{k=1}%
^{K}\sum_{l=1}^{L_{C}}\mathbf{Q}_{kl}\right)  \leq P_{tot}$.
\end{itemize}

\section{Problem Solution via Branch-and-Bound\label{Sec_algorithm}}

In this section we show that, adopting standard BB techniques (see \cite{bib:
boyd 1991} and\ \cite{bib: tuy decomposition}), problem $\left(
\mathrm{P}\right)  $ can be optimally solved by means of an efficient BB that
exploits the structure dictated by assumptions (A1-A3). The BB algorithm is
fully characterized\ by two procedures: \textit{branching} and
\textit{bounding}. These are iteratively performed until the solution's
suboptimality falls below some prescribed accuracy $\varepsilon$. In the
following we explicitly tailor those procedures to problem $\left(
\mathrm{P}\right)  $ and we show convergence of the proposed BB algorithm to
the global optimal solution of $\left(  \mathrm{P}\right)  $. For
readability's sake, we define $\mathbf{Q}\colon=\left\{  \mathbf{Q}%
_{k}\right\}  $\footnote{Here and in the following, the expression
$\mathbf{Q}\in\mathcal{Q}$ stands for $\left\{  \mathbf{Q}_{k}\right\}
\in\mathcal{Q}$.} and we address an interval as $\mathcal{M}\colon=\left[
\mathbf{a},\mathbf{b}\right]  $, meaning that $\mathbf{c}\in\mathcal{M}%
\Leftrightarrow\left[  \mathbf{a}\right]  _{l}\leq\left[  \mathbf{\mathbf{c}%
}\right]  _{l}\leq\left[  \mathbf{b}\right]  _{l}$ for $l=1,\ldots,L$.

\subsection{Branching Procedure}

A partition set $\mathcal{P}_{t}$ of rectangles $\left\{  \mathcal{M}\right\}
$ in the space $\mathbb{R}^{L}$, each labeled with a \textit{lower}
$\mathfrak{L}_{\mathfrak{B}}\left(  \mathcal{M}\right)  $ and \textit{upper}
$\mathfrak{U}_{\mathfrak{B}}\left(  \mathcal{M}\right)  $ bounds, is given. By
splitting a rectangle that satisfies $\mathcal{M}_{t}\in\arg\min
_{\mathcal{M}\in\mathcal{P}_{t}}\mathfrak{L}_{\mathfrak{B}}\left(
\mathcal{M}\right)  $ in $J$ non-overlapping sub-rectangles $\left\{
\mathcal{\hat{M}}_{t}\right\}  $ (i.e., $\bigcap\nolimits_{j=1}^{J}%
\mathcal{\hat{M}}_{t}^{\left(  j\right)  }=\mathcal{\emptyset}$ and
$\bigcup\nolimits_{j=1}^{J}\mathcal{\hat{M}}_{t}^{\left(  j\right)
}=\mathcal{M}_{t}$), the enhanced partition $\mathcal{P}_{t+1}\triangleq
\left\{  \mathcal{P}_{t}\backslash\mathcal{M}_{t}\right\}  \cup\left\{
\mathcal{\hat{M}}_{t}\right\}  $ is obtained. Lower and upper bounds for each
sub-rectangle in $\left\{  \mathcal{\hat{M}}_{t}\right\}  $ are then obtained
via the following \textit{bounding procedure}.

\subsection{Bounding Procedure}

Exploiting the \textit{Partly Convex-Monotone} structure of problem $\left(
\mathrm{P}\right)  $, for every rectangle $\mathcal{M}=\left[  \mathbf{i}%
^{\min},\mathbf{i}^{\max}\right]  \in\mathcal{P}_{t}$, a \textit{lower bound}
$\mathfrak{L}_{\mathfrak{B}}\left(  \mathcal{M}\right)  $ is evaluated by
solving the following problem:%
\begin{align}
\mathfrak{L}_{\mathfrak{B}}\left(  \mathcal{M}\right)  \colon &  =\text{
}\underset{\mathbf{Q}\in\mathcal{Q}}{\min}f\left(  \mathbf{Q},\mathbf{i}%
^{\min}\right) \label{eq: Pbeta}\\
s.t.\text{\ } &  \mathbf{i}^{\min}\preccurlyeq\mathbf{\mathbf{\mathbf{f}%
}_{\mathbf{i}}\left(  \mathbf{Q}\right)  }\preccurlyeq\mathbf{i}^{\max
}\text{.}\nonumber
\end{align}
Thanks to assumptions (A1-A3) two fundamental results can be
verified:\ (\textit{i}) problem (\ref{eq: Pbeta}) is \textit{convex} since the
cost function $f\left(  \mathbf{Q},\mathbf{i}\right)  $ is convex for a fixed
$\mathbf{i}$ and the constraints form a convex set, (\textit{ii}) using
standard convex optimization arguments, it can be shown that this bounding
procedure satisfies the natural condition:\vspace{-0.07cm}%
\begin{equation}
\mathcal{M}^{\prime}\subset\mathcal{M}\Rightarrow\mathfrak{L}_{\mathfrak{B}%
}\left(  \mathcal{M}^{\prime}\right)  \geq\mathfrak{L}_{\mathfrak{B}}\left(
\mathcal{M}\right)  .\label{eq: bound condition}%
\end{equation}
Moreover, denoting with $\mathbf{Q}^{\left(  \mathfrak{L}_{\mathfrak{B}%
}\right)  }$ the optimal solution of problem (\ref{eq: Pbeta}), a valid upper
bound $\mathfrak{U}_{\mathfrak{B}}\left(  \mathcal{M}\right)  $ is obtained by
evaluating the function at $\mathbf{Q}^{\left(  \mathfrak{L}_{\mathfrak{B}%
}\right)  }$, i.e., $\mathfrak{U}_{\mathfrak{B}}\left(  \mathcal{M}\right)
\colon=f\left(  \mathbf{Q}^{\left(  \mathfrak{L}_{\mathfrak{B}}\right)
},\mathbf{\mathbf{\mathbf{f}}_{\mathbf{i}}}\left(  \mathbf{Q}^{\left(
\mathfrak{L}_{\mathfrak{B}}\right)  }\right)  \right)  $.\newline Finally, the
algorithm checks if the prescribed accuracy is met (i.e., if $\min
\mathfrak{U}_{\mathfrak{B}}\left(  \mathcal{M}\right)  -\min\mathfrak{L}%
_{\mathfrak{B}}\left(  \mathcal{M}\right)  \leq\varepsilon$) otherwise it goes
back to the \textit{branching procedure}.

\subsection{Convergence Analysis}

Here we proves convergence of the proposed BB algorithm.

\begin{lemma}
\label{lemma 1} The proposed BB algorithm (which is performed in the
reduced space spanned by interference levels/variable $\mathbf{i}$), is
\textit{convergent} to a global optimal solution of problem $\left( \mathrm{P%
}\right) $.
\end{lemma}

\begin{proof}
As explained above, since the chosen bounding procedure satisfyies (\ref{eq:
bound condition}), the BB algorithm generates a sequences of partition
sets $\left\{ \mathcal{M}_{t}\right\} $ collapsing to a point $%
\bigcap_{t\rightarrow \infty }\mathcal{M}_{t}=\mathbf{i}^{\mathbf{\ast }}$
(recall that $\mathcal{M}_{t}$ is the rectangle selected for splitting at
the $t$-th branching iteration). In order to prove convergence we need to
show that, as the size of rectangle $\mathcal{M}_{t}$ gets smaller, $%
\mathfrak{U}_{\mathfrak{B}}\left( \mathcal{M}_{t}\right) -\mathfrak{L}_{%
\mathfrak{B}}\left( \mathcal{M}_{t}\right) $ is also sufficiently small. The
proof follows standard arguments \cite{bib: boyd 1991}. This is shown in
Appendix.
\end{proof}

\subsection{Broadcast WSRM Example}

Considering the BC WSRM scenario (see Sec.\ref{marker: example Utility}), for
a given interval $\mathcal{M}=\left[  \mathbf{i}^{\min},\mathbf{i}^{\max
}\right]  $, the evaluation of a lower bound $\mathfrak{L}_{\mathfrak{B}}$
results in the following convex problem%
\begin{align}
\mathfrak{L}_{\mathfrak{B}}\left(  \mathcal{M}\right)  \colon &  =\text{
}\underset{\left\{  \mathbf{Q}_{k}\succcurlyeq\mathbf{0}\right\}  }{\min}%
\sum_{k=1}^{K}-w_{k}\log\left(  1\mathbf{+}\frac{\mathbf{\mathbf{h}}%
_{k}\mathbf{Q}_{k}\mathbf{\mathbf{h}}_{k}^{H}}{\sigma_{k}^{2}\mathbf{+}%
i_{k}^{\min}}\right) \label{eq: LB}\\
s.t.\text{ } &  \left\{
\begin{array}
[c]{lc}%
\operatorname{Tr}\left(  \sum_{k=1}^{K}\mathbf{Q}_{k}\right)  \leq P_{tot} &
\\
i_{k}^{\min}\leq\mathbf{\mathbf{h}}_{k}\left(
{\textstyle\sum\nolimits_{j\neq k}}
\mathbf{Q}_{j}\right)  \mathbf{\mathbf{h}}_{k}^{H}\leq i_{k}^{\max} & \forall
k
\end{array}
\text{.}\right. \nonumber
\end{align}
Defining $\left\{  \mathbf{Q}_{k}^{\ast}\right\}  $ the optimal solution of
(\ref{eq: LB}),\ a valid upper bound is given by $\mathfrak{U}_{\mathfrak{B}%
}\left(  \mathcal{M}\right)  \colon=\sum_{k=1}^{K}-w_{k}\log\left(
1\mathbf{+}\frac{\mathbf{\mathbf{h}}_{k}\mathbf{Q}_{k}^{\ast}%
\mathbf{\mathbf{h}}_{k}^{H}}{\sigma_{k}^{2}\mathbf{+}i_{k}}\right)  $ where
$i_{k}=\mathbf{\mathbf{h}}_{k}\left(
{\textstyle\sum\nolimits_{j\neq k}}
\mathbf{Q}_{j}^{\ast}\right)  \mathbf{\mathbf{h}}_{k}^{H}$ $\forall k$.

\section{Suboptimal Solution\label{Sec_algorithm_DD}}

While the proposed BB algorithm always converges to the global optimal
solution and has reduced complexity with respect to a general-purpose
implementation of BB, it is still feasible only for offline simulation. In
this section we propose a suboptimal algorithm with polynomial complexity that
extends the distributed pricing schemes of \cite{bib: berry honig pricing
MISO} to the more general class of problem $\left(  \mathrm{P}\right)
$.\newline Exploiting the \textit{Partly Convex-Monotone} structure, problem
(\ref{eq: Problem_auxiliar_i}) can be equivalently reformulated as the
non-convex problem:%
\begin{equation}
\min_{\mathbf{i}\in\mathcal{M}_{0}}\sup_{\mathbf{\lambda}}\left[
\underset{\mathbf{Q}\in\mathcal{Q}}{\min}\left[  f\left(  \mathbf{Q}%
,\mathbf{i}\right)  +\mathbf{\lambda}^{T}\mathbf{\mathbf{f}}_{\mathbf{i}%
}\left(  \mathbf{Q}\right)  \right]  -\mathbf{\lambda}^{T}\mathbf{i}\right]
\text{.}\label{eq: convex opt g}%
\end{equation}
where $\mathbf{\lambda}$\ is the Lagrange multiplier associated to the affine
constraint $\mathbf{i}=\mathbf{\mathbf{f}}_{\mathbf{i}}\left(  \mathbf{Q}%
\right)  $. Building on (\ref{eq: convex opt g}), in the following table we
formalize the proposed suboptimal algorithm.

\begin{algorithm} \caption{\ -
Suboptimally solve problem $\left( \mathrm{P}\right) $} \label{algo1}
\emph{0}: \ \ \emph{Set} $\varepsilon _{\mathbf{\lambda }},\varepsilon _{%
\mathbf{i}}$
\newline
\emph{1}: \ \ \emph{Initialize} $\mathbf{\lambda}=\widehat{\mathbf{\lambda }%
}$
\newline
\emph{2}: \ \ \emph{Initialize} $\mathbf{i}=\widehat{\mathbf{i}}$
\newline
\emph{3}: \ \ \emph{Evaluate }$\mathbf{Q}^{\mathbf{\ast }}=\arg \underset{%
\mathbf{Q}\in \mathcal{Q}}{\min }f\left( \mathbf{Q},\widehat{\mathbf{i}}%
\right) +\widehat{\mathbf{\lambda }}^{T}\mathbf{f}_{\mathbf{i}}\left(
\mathbf{Q}\right)$
\newline
\emph{4: \ \ }\textbf{If }$\left\Vert \widehat{\mathbf{i}}-\mathbf{f}_{%
\mathbf{i}}\left( \mathbf{Q}^{\mathbf{\ast }}\right) \right\Vert
>\varepsilon _{\mathbf{i}}$
\newline
\emph{5: \ \ }\qquad \emph{Update} $\widehat{\mathbf{i}}=\mathbf{f}_{\mathbf{i}}\left( \mathbf{Q}^{\mathbf{\ast }}\right) $
\newline
\emph{6: \ \ }\qquad \emph{Go back to step 3}
\newline
\emph{7}: \ \textbf{elseIf }$\left\Vert \widehat{\mathbf{\lambda }}-\left.
\frac{\partial f\left( \mathbf{Q},\mathbf{i}\right) }{\partial \mathbf{i}}%
\right\vert _{\mathbf{Q}=\mathbf{Q}^{\ast },\mathbf{i}=\widehat{\mathbf{i}}%
}\right\Vert >\varepsilon _{\mathbf{\lambda }}$
\newline
\emph{8: \ \ }\qquad \emph{Update} $\widehat{\mathbf{\lambda }}=\left. \frac{%
\partial f\left( \mathbf{Q},\mathbf{i}\right) }{\partial \mathbf{i}}%
\right\vert _{\mathbf{Q}=\mathbf{Q}^{\ast },\mathbf{i}=\widehat{\mathbf{i}}}$
\newline
\emph{9:\ \ \ }\qquad \emph{Go back to step 2}
\newline
\emph{10}: \textbf{end}
\end{algorithm}

$ $\newline
Since a stationary point of this algorithm fulfills
the \textit{necessary} Karush-Kuhn-Tucker (KKT) conditions of problem
(\ref{eq: Problem_auxiliar_i}), if the algorithm converges, it attains a local
optimal point of problem\ (\ref{eq: Problem_auxiliar_i}).\newline It is worth
noticing that, by specializing our framework to the case when the cost
function $f\left(  \mathbf{Q},\mathbf{i}\right)  $ is the WSRM (i.e.,
$f_{\alpha}\left(  r_{k}\right)  =r_{k}$ $\forall k$ in (\ref{eq: IC WSRM}))
the Lagrangian multiplier $\mathbf{\lambda}$ plays the role of the
\textit{interference prices} defined in the distributed pricing algorithm
\cite{bib: berry honig pricing MISO}. Thus Algorithm \ref{algo1} can be seen
as a generalization of distributed pricing technique with an arbitrarily cost
function and arbitrary \textit{interference} functions
(satisfying\ assumptions A1-A3).\newline Finally, since the problem at hand is
non-convex, initialization of the parameters $\mathbf{\lambda}$ and
$\mathbf{i}$ results crucial for performances and convergence. In
Sec.\ref{sec:num res} we assess the performances of this technique in relation
to the global optimal solution evaluated via BB algorithm.

\section{Numerical Results}

\label{sec:num res}

We assess the performance of the two proposed techniques: a multi-user linear
precoder optimized via (\textit{i}) the efficient BB algorithm (BB - LB and
UB); (\textit{ii}) the suboptimal Algorithm \ref{algo1}.\newline We consider
the sum-rate utility function (i.e., in (\ref{eq: IC WSRM}), $f_{\alpha
}\left(  r_{k}\right)  =r_{k}$ $\forall k$ and $w_{k}=1$ $\forall k$). In BB
algorithm, the solution's accuracy is $\varepsilon=10^{-3}$, while, in
Algorithm \ref{algo1}, we run two different price initializations
($\lambda_{k}=10^{-5}$ $\forall k$ and $\lambda_{k}=1$ $\forall k$) and for
both we initialize $i_{k}=\sigma_{k}^{2}$ $\forall k$, selecting $\sigma
_{k}^{2}=1$ at each receive antenna.%
\begin{figure}
[ptb]
\begin{center}
\includegraphics[
trim=0.000000in 0.000000in -0.014697in 0.000000in,
height=2.5598in,
width=3.352in
]%
{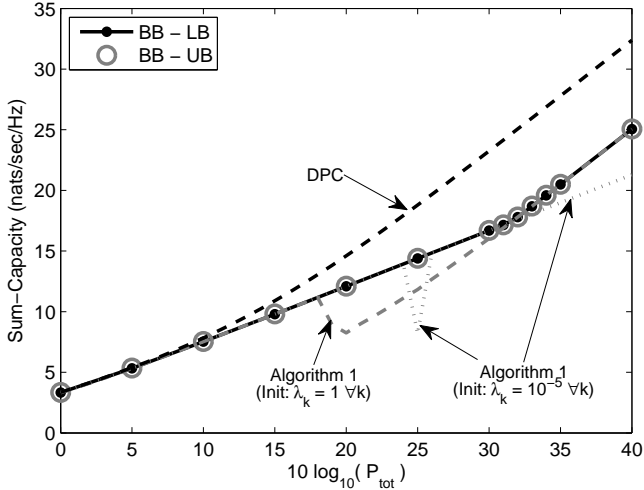}%
\caption{Sum-Rate versus transmitting power for the single-carrier ($L_{C}=1$)
BC scenario. The figure compares the two proposed linear precoding algorithms:
the optimal BB algorithm (BB - LB and UB) and the suboptimal Algorithm
\ref{algo1} considering two different initializations ($\lambda_{k}=10^{-5}$
$\forall k$ and $\lambda_{k}=1$ $\forall k$). The optimal \textit{non-linear}
DPC technique is also plotted as a reference.}%
\label{fig. 1}%
\end{center}
\end{figure}
\newline Fig.\ref{fig. 1} shows the sum-rate versus the transmitting power for
a single-carrier ($L_{C}=1$) BC channel\footnote{Due to space limitation,
Fig.\ref{fig. 1} channels realization is available at:
\texttt{http://web.njit.edu/\symbol{126}mr227/papers/paper\_BB\_H\_BC.mat}}
where a $N=4$ transmit antennas base-station serves $K=4$ single-antennas
users, subject to a sum-power constraint, $\operatorname{Tr}\left(  \sum
_{k=1}^{K}\mathbf{Q}_{k}\right)  \leq P_{tot}$. The sum-capacity achieving
\textit{non-linear} technique Dirty Paper Coding (DPC) is also plotted as a
reference.\newline It can be noticed that the suboptimal Algorithm
\ref{algo1}, while showing near-optimal performance at several power levels,
happens to be quite sensitive to initialization. For instance, initialization
$\lambda_{k}=10^{-5}$ $\forall k$ yields a suboptimal slope in\ high power
regime, as observed for\ $\left.  P_{tot}\right\vert _{dB}>31$, and, at
$\left.  P_{tot}\right\vert _{dB}=25dB$ both initializations lead to highly
suboptimal performances. A last observation pertains to the significant gains
of \textit{non-linear} DPC with respect to \textit{linear} precoding at high
power regime.\newline Finally, not to confuse the reader, since a utility
function (sum-rate) instead of a cost function is plotted, in fig.\ref{fig. 1}%
, the lower bound results as the maximum feasible value while the upper bound
is the maximum upper bound among BB partitions.

\section{Conclusions}

This work presents a global optimization framework for the minimization of
non-convex cost functions in MISO BC and IC channels. Examples are given for
the general $\alpha$-fairness optimization considering parallel IC and BC
channels. Knowing the global optimal solution, even if impractical for
real-time implementation, allows to assess the quality and to fine-tune (e.g.,
initialize) suboptimal schemes. In addition to the global optimal BB, we have
proposed a real-time, hence suboptimal, algorithm that generalizes the pricing
scheme of \cite{bib: berry honig pricing MISO}. Extensions to MIMO networks
are the subject of future work.

\section{Appendix}

We need to prove that, as the maximum length of the edges of $\mathcal{M}_{t}
$, denoted by $size(\mathcal{M}_{t})$, goes to zero, the difference between
upper and lower bounds uniformly converges to zero, i.e., $\forall
\varepsilon>0\ \exists\delta>0\ \forall\mathcal{M}_{t}\subseteq\mathcal{M}%
_{0}$ $\ size(\mathcal{M}_{t})\leq\delta\Longrightarrow\mathfrak{U}%
_{\mathfrak{B}}(\mathcal{M}_{t})-\mathfrak{L}_{\mathfrak{B}}(\mathcal{M}%
_{t})\leq\varepsilon$.\newline For each $\widehat{\mathbf{i}}\in
\mathcal{M}_{t}=\left[  \mathbf{i}^{\min},\mathbf{i}^{\max}\right]  $, we
define the function $\mathfrak{F}\left(  \widehat{\mathbf{i}}\right)  $ as the
result of the following constraint optimization problem:%
\begin{gather*}
\underset{\mathbf{Q}\in\mathcal{Q}}{\min}f\left(  \mathbf{Q},\widehat
{\mathbf{i}}\right) \\
s.t.\text{ \ }\mathbf{i}^{\min}\preccurlyeq\mathbf{\mathbf{\mathbf{f}%
}_{\mathbf{i}}\left(  \mathbf{Q}\right)  }\preccurlyeq\mathbf{i}^{\max
}\text{.}%
\end{gather*}
Using this notation, the lower bound in (\ref{eq: Pbeta}) is given by
$\mathfrak{L}_{\mathfrak{B}}=\mathfrak{F}\left(  \mathbf{i}^{\min}\right)  $,
while an upper bound is given by $\mathfrak{U}_{\mathfrak{B}}=\mathfrak{F}%
\left(  \mathbf{i}^{\max}\right)  $.\newline From jointly-continuity of the
function $f\left(  \mathbf{Q},\widehat{\mathbf{i}}\right)  $ with respect to
$\left(  \mathbf{Q},\widehat{\mathbf{i}}\right)  $ (assumption A2) and from
the definition of $\mathfrak{F}\left(  \widehat{\mathbf{i}}\right)  $, we have
that $\mathfrak{F}\left(  \widehat{\mathbf{i}}\right)  $ is continuous in the
norm of $\widehat{\mathbf{i}}$ (i.e., $\left\Vert \widehat{\mathbf{i}%
}\right\Vert $). It follows that also $\mathfrak{L}_{\mathfrak{B}}$ and
$\mathfrak{U}_{\mathfrak{B}}$ will result continue in $\left\Vert
\widehat{\mathbf{i}}\right\Vert $, thus it holds\vspace{-0.15cm}%
\[
\forall\varepsilon\text{ }\exists\delta\ \ \ \left\Vert \mathbf{i}^{\max
}-\mathbf{i}^{\min}\right\Vert \leq\delta\Longrightarrow\left\vert
\mathfrak{F}\left(  \mathbf{i}^{\max}\right)  -\mathfrak{F}\left(
\mathbf{i}^{\min}\right)  \right\vert \leq\varepsilon
\]
\vspace{-0.3cm}concluding the proof.

\end{document}